\documentclass[runningheads,envcountsame]{llncs}

\usepackage{xspace}
\usepackage[T1]{fontenc}
\usepackage[noadjust]{cite}
\usepackage{bbm,mathtools,amssymb,bbold}
\usepackage[heavycircles]{stmaryrd}
\usepackage{paralist}
\usepackage{wrapfig}
\usepackage{url}
\usepackage{tikz}
\usetikzlibrary{decorations}
\usetikzlibrary{automata}
\usetikzlibrary{shapes.multipart}
\usetikzlibrary{arrows}

\makeatletter
\def\xrightarrowtrianglefill@{\arrowfill@\relbar\relbar\rightarrowtriangle}
\providecommand{\xrightarrowtriangle}[2][]{%
  \ext@arrow 0099\xrightarrowtrianglefill@{#1}{#2}}
\makeatother
\newcommand*\cf{\textit{cf.}}
\newcommand*\etc{\textit{etc.}}
\newcommand*\ie{\textit{i.e.},}
\newcommand*\Real{\mathbbm{R}}
\newcommand*\Realnn{\Real_{ \ge 0}}
\newcommand*\Realni{\Real_*}
\newcommand*\tto[1]{\xrightarrow{#1}}
\newcommand*\mcal[1]{\mathcal #1}
\newcommand*\conf{\textup{\textsf{Conf}}}
\newcommand*\pend{\textup{\textsf{end}}}
\newcommand*\val{\textup{\textsf{val}}}
\newcommand*\ato[1]{\xrightarrowtriangle{#1}}
\newcommand*\mato{\mathord{\ato{}}}
\newcommand*\cpre{\textup{\textsf{cpre}}}
\newcommand*\sco{\textup{\textsf{sc}}}
\newcommand*\Nat{\mathbbm{N}}

\title{Computing Branching Distances \\ Using Quantitative Games}

\author{Uli Fahrenberg\inst1\thanks{This author's work is supported by
    the \textit{Chaire ISC~: Engineering Complex Systems} --
    \smash{\'E}cole polytechnique -- Thales -- FX -- DGA -- Dassault
    Aviation -- DCNS Research -- ENSTA ParisTech -- T{\'e}l{\'e}com
    ParisTech} \and Axel Legay\inst2\inst3 \and Karin Quaas\inst4}

\institute{{\'E}cole polytechnique, Palaiseau, France \and
  Universit{\'e} catholique de Louvain, Belgium \and Aalborg
  University, Denmark \and Universit{\"a}t Leipzig, Germany}

\begin{document}

\maketitle

\begin{abstract}
  We lay out a general method for computing branching distances
  between labeled transition systems.  We translate the quantitative
  games used for defining these distances to other, path-building
  games which are amenable to methods from the theory of quantitative
  games.  We then show for all common types of branching distances how
  the resulting path-building games can be solved.  In the end, we
  achieve a method which can be used to compute all branching
  distances in the linear-time--branching-time spectrum.

\end{abstract}

\section{Introduction}

During the last decade, formal verification has seen a trend towards
modeling and analyzing systems which contain quantitative information.
This is motivated by applications in real-time systems, hybrid systems,
embedded systems and others.  Quantitative information can thus be a
variety of things: probabilities, time, tank pressure, energy intake,
\etc

A number of quantitative models have hence been developed:
probabilistic automata~\cite{DBLP:conf/concur/SegalaL94}, stochastic
process algebras~\cite{book/Hillston96}, timed
automata~\cite{DBLP:journals/tcs/AlurD94}, hybrid
automata~\cite{DBLP:journals/tcs/AlurCHHHNOSY95}, timed variants of
Petri nets~\cite{journal/transcom/MerlinF76, DBLP:conf/apn/Hanisch93},
continuous-time Markov chains~\cite{book/Stewart94}, \etc\; Similarly,
there is a number of specification formalisms for expressing
quantitative properties: timed computation tree
logic~\cite{DBLP:journals/iandc/HenzingerNSY94}, probabilistic
computation tree logic~\cite{DBLP:journals/fac/HanssonJ94}, metric
temporal logic~\cite{DBLP:journals/rts/Koymans90}, stochastic
continuous logic~\cite{DBLP:journals/tocl/AzizSSB00}, \etc

Quantitative verification, \ie~the checking of quantitative properties
for quantitative systems, has also seen rapid development: for
probabilistic systems in PRISM~\cite{DBLP:conf/tacas/KwiatkowskaNP02}
and PEPA~\cite{DBLP:conf/cpe/GilmoreH94}, for real-time systems in
Uppaal~\cite{DBLP:journals/sttt/LarsenPY97},
RED~\cite{DBLP:conf/fm/WangME93}, TAPAAL~\cite{DBLP:conf/atva/BygJS09}
and Romeo~\cite{DBLP:conf/cav/GardeyLMR05}, and for hybrid systems in
HyTech~\cite{DBLP:journals/sttt/HenzingerHW97},
SpaceEx~\cite{DBLP:conf/cav/FrehseGDCRLRGDM11} and
HySAT~\cite{DBLP:journals/fmsd/FranzleH07}, to name but a few.

Quantitative verification has, however, a problem of
\emph{robustness}.  When the answers to model checking problems are
Boolean---either a system meets its specification or it does
not---then small perturbations in the system's parameters may
invalidate the result.  This means that, from a model checking point
of view, small, perhaps unimportant, deviations in quantities are
indistinguishable from larger ones which may be critical.

As an example, Fig.~\ref{fi:tatrain} shows three simple
timed-automaton models of a train crossing, each modeling that once
the gates are closed, some time will pass before the train arrives.
Now assume that the specification of the system is
\begin{equation*}
  \textit{The gates have to be closed 60 seconds before the train
    arrives.}
\end{equation*}
Model $A$ does guarantee this property, hence satisfies the
specification.  Model $B$ only guarantees that the gates are closed 58
seconds before the train arrives, and in model $C$, only one second
may pass between the gates closing and the train.

Neither of models $B$ and $C$ satisfies the specification, so this is
the result which a model checker like for example Uppaal would output.
What this does not tell us, however, is that model $C$ is dangerously
far away from the specification, whereas model $B$ only violates it
slightly (and may be acceptable from a practical point of view given
other constraints on the system which we have not modeled here).

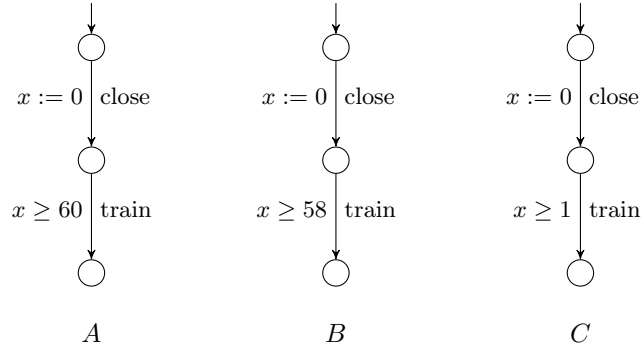
\begin{figure}[tbp]
  \centering
  \begin{tikzpicture}[->,>=stealth',auto,initial text=]
    \tikzstyle{every node}=[font=\small]
    \tikzstyle{every state}=[inner sep=.5mm,minimum size=3.5mm]
    \begin{scope}
      \node[state,initial above] (0) at (0,0) {};
      \node[state] (1) at (0,-1.5) {};
      \node[state] (2) at (0,-3) {};
      \path (0) edge node [left,anchor=base east] {$x:= 0$} node
      [right,anchor=base west] {close} (1);
      \path (1) edge node [left,anchor=base east] {$x\ge 60$} node
      [right,anchor=base west] {train} (2);
      \node[font=\normalsize] at (0,-3.8) {$A$};
    \end{scope}
    \begin{scope}[xshift=10em]
      \node[state,initial above] (0) at (0,0) {};
      \node[state] (1) at (0,-1.5) {};
      \node[state] (2) at (0,-3) {};
      \path (0) edge node [left,anchor=base east] {$x:= 0$} node
      [right,anchor=base west] {close} (1);
      \path (1) edge node [left,anchor=base east] {$x\ge 58$} node
      [right,anchor=base west] {train} (2);
      \node[font=\normalsize] at (0,-3.8) {$B$};
    \end{scope}
    \begin{scope}[xshift=20em]
      \node[state,initial above] (0) at (0,0) {};
      \node[state] (1) at (0,-1.5) {};
      \node[state] (2) at (0,-3) {};
      \path (0) edge node [left,anchor=base east] {$x:= 0$} node
      [right,anchor=base west] {close} (1);
      \path (1) edge node [left,anchor=base east] {$x\ge 1$} node
      [right,anchor=base west] {train} (2);
      \node[font=\normalsize] at (0,-3.8) {$C$};
    \end{scope}
  \end{tikzpicture}
  \caption{%
    \label{fi:tatrain}
    Three timed automata modeling a train crossing.
  }
\end{figure}

In order to address the robustness problem, one approach is to replace
the Boolean yes-no answers of standard verification with distances.
That is, the Boolean co-domain of model checking is replaced by the
non-negative real numbers.  In this setting, the Boolean \texttt{true}
corresponds to a distance of zero and \texttt{false} to the non-zero
numbers, so that quantitative model checking can now tell us not only
that a specification is violated, but also \emph{how much} it is
violated, or \emph{how far} the system is from corresponding to its
specification.

In the example of Fig.~\ref{fi:tatrain}, and depending on precisely
how one wishes to measure distances, the distance from $A$ to our
specification would be $0$, whereas the distances from $B$ and $C$ to
the specification may be $2$ and $59$, for example.  The precise
interpretation of distance values will be application-dependent; but
in any case, it is clear that $C$ is much farther away from the
specification than $B$ is.

The distance-based approach to quantitative verification has been
developed in~\cite{DBLP:journals/tcs/DesharnaisGJP04,
  DBLP:journals/tcs/BreugelW05, DBLP:journals/tcs/AlfaroFHMS05,
  DBLP:conf/formats/HenzingerMP05, DBLP:journals/tac/GirardP07,
  DBLP:journals/tcs/Breugel01, DBLP:journals/jlp/ThraneFL10} and many
other papers.  Common to all these approaches is that they introduce
distances between systems, or between systems and specifications, and
then employ these for approximate or quantitative verification.
However, depending on the application context, a plethora of different
distances are being used.  Consequently, there is a need for a general
theory of quantitative verification which depends as little as
possible on the concrete distances being used.

Different applications foster different types of quantitative
verification, but it turns out that most of these essentially measure
some type of distances between labeled transition systems.  We have
in~\cite{DBLP:journals/tcs/FahrenbergL14} laid out a unifying
framework which allows one to reason about such distance-based
quantitative verification independently of the precise distance.  This
is essentially a general metric theory of labeled transition systems,
with infinite quantitative games as its main theoretical ingredient
and general fixed-point equations for linear and branching distances
as one of its main results.

The work in~\cite{DBLP:journals/tcs/FahrenbergL14} generalizes the
linear-time--branching-time spectrum of preorders and equivalences
from van~Glabbeek's~\cite{inbook/hpa/Glabbeek01} to a quantitative
linear-time--branching-time spectrum of distances, all parameterized
on a given distance on traces, or executions;
\cf~Fig.~\ref{fi:spectrum}.  This is done by generalizing Stirling's
bisimulation game~\cite{DBLP:conf/banff/Stirling95} along two
directions, both to cover all other preorders and equivalences in the
linear-time--branching-time spectrum and into a game with quantitative
(instead of Boolean) objectives.

\begin{figure}[p]
  \centering
  \begin{tikzpicture}[->,xscale=.85,yscale=1.1]
    \tikzstyle{every node}=[font=\small,text badly centered]
    \node (traceeq) at (0,.3) {$\infty$-nested trace equivalence};
    \node (k+1-r-trace) at (-2,-2.4) {$( k+ 1)$-nested ready
      inclusion};
    \node (k+1-traceeq) at (2,-1.6) {$( k+ 1)$-nested trace equivalence};
    \node (k-r-traceeq) at (-2,-4.4) {$k$-nested ready equivalence};
    \node (k+1-trace) at (2,-3.6) {$( k+ 1)$-nested trace inclusion};
    \node (k-r-trace) at (-2,-6.4) {$k$-nested ready inclusion};
    \node (k-traceeq) at (2,-5.6) {$k$-nested trace equivalence};
    \node (2-r-trace) at (-2,-9.4) {$2$-nested ready inclusion};
    \node [text width=11.6em] (2-traceeq) at (2,-8.6) {$2$-nested trace
      equivalence \\ \emph{possible-futures equivalence}};
    \node [text width=11.4em] (1-r-traceeq) at (-2,-11.4) {$1$-nested ready
      equivalence \\ \emph{ready equivalence}};
    \node [text width=11.5em] (2-trace) at (2,-10.6) {$2$-nested trace
      inclusion \\ \emph{possible-futures inclusion}};
    \node [text width=11em] (1-r-trace) at (-2,-13.4) {$1$-nested ready
      inclusion \\ \emph{ready inclusion}};
    \node [text width=11.9em] (1-traceeq) at (2,-12.6) {$1$-nested trace
      equivalence \\ \emph{trace equivalence}};
    \node [text width=10.7em] (1-trace) at (2,-14.6) {$1$-nested trace
      inclusion \\ \emph{trace inclusion}};
    \node [text width=16em] (bisim) at (7,.3) {$\infty$-nested
      simulation equivalence \\ \emph{bisimulation}};
    \node (k+1-r-sim) at (5,-2.4) {$( k+ 1)$-ready sim.~equivalence};
    \node [text width=9em] (k+1-simeq) at (9,-1.6) {$( k+ 1)$-nested
      sim. equivalence};
    \node (k-r-simeq) at (5,-4.4) {$k$-nested ready sim.~equivalence};
    \node (k+1-sim) at (9,-3.6) {$( k+ 1)$-nested simulation};
    \node (k-r-sim) at (5,-6.4) {$k$-nested ready simulation};
    \node (k-simeq) at (9,-5.6) {$k$-nested sim.~equivalence};
    \node (2-r-sim) at (5,-9.4) {$2$-nested ready simulation};
    \node (2-simeq) at (9,-8.6) {$2$-nested sim.~equivalence};
    \node [text width=14.5em] (1-r-simeq) at (5,-11.4) {$1$-nested ready
      sim.~equivalence \\ \emph{ready simulation equivalence}};
    \node (2-sim) at (9,-10.6) {$2$-nested simulation};
    \node [text width=12em] (1-r-sim) at (5,-13.4) {$1$-nested ready
      simulation \\ \emph{ready simulation}};
    \node [text width=12em] (1-simeq) at (9,-12.6) {$1$-nested
      sim.~equivalence \\ \emph{simulation equivalence}};
    \node [text width=10em] (1-sim) at (9,-14.6) {$1$-nested
      simulation \\ \emph{simulation}};
    %
    \path (bisim) edge (traceeq);
    \path (k+1-r-sim) edge (k+1-r-trace);
    \path (k+1-simeq) edge (k+1-traceeq);
    \path (k-r-simeq) edge (k-r-traceeq);
    \path (k+1-sim) edge (k+1-trace);
    \path (k-r-sim) edge (k-r-trace);
    \path (k-simeq) edge (k-traceeq);
    \path (2-simeq) edge (2-traceeq);
    \path (2-r-sim) edge (2-r-trace);
    \path (2-sim) edge (2-trace);
    \path (1-r-simeq) edge (1-r-traceeq);
    \path (1-simeq) edge (1-traceeq);
    \path (1-r-sim) edge (1-r-trace);
    \path (1-sim) edge (1-trace);
    \path [dashed] (traceeq) edge (k+1-r-trace);
    \path [dashed] (traceeq) edge (k+1-traceeq);
    \path (k+1-r-trace) edge (k-r-traceeq);
    \path (k+1-r-trace) edge (k+1-trace);
    \path (k+1-traceeq) edge (k-r-traceeq);
    \path (k+1-traceeq) edge (k+1-trace);
    \path (k-r-traceeq) edge (k-r-trace);
    \path (k-r-traceeq) edge (k-traceeq);
    \path (k+1-trace) edge (k-r-trace);
    \path (k+1-trace) edge (k-traceeq);
    \path [dashed] (k-r-trace) edge (2-r-trace);
    \path [dashed] (k-r-trace) edge (2-traceeq);
    \path [dashed] (k-traceeq) edge (2-r-trace);
    \path [dashed] (k-traceeq) edge (2-traceeq);
    \path (2-r-trace) edge (1-r-traceeq);
    \path (2-r-trace) edge (2-trace);
    \path (2-traceeq) edge (1-r-traceeq);
    \path (2-traceeq) edge (2-trace);
    \path (1-r-traceeq) edge (1-r-trace);
    \path (1-r-traceeq) edge (1-traceeq);
    \path (2-trace) edge (1-r-trace);
    \path (2-trace) edge (1-traceeq);
    \path (1-r-trace) edge (1-trace);
    \path (1-traceeq) edge (1-trace);
    \path [dashed] (bisim) edge (k+1-r-sim);
    \path [dashed] (bisim) edge (k+1-simeq);
    \path (k+1-r-sim) edge (k-r-simeq);
    \path (k+1-r-sim) edge (k+1-sim);
    \path (k+1-simeq) edge (k-r-simeq);
    \path (k+1-simeq) edge (k+1-sim);
    \path (k-r-simeq) edge (k-r-sim);
    \path (k-r-simeq) edge (k-simeq);
    \path (k+1-sim) edge (k-r-sim);
    \path (k+1-sim) edge (k-simeq);
    \path [dashed] (k-r-sim) edge (2-r-sim);
    \path [dashed] (k-r-sim) edge (2-simeq);
    \path [dashed] (k-simeq) edge (2-r-sim);
    \path [dashed] (k-simeq) edge (2-simeq);
    \path (2-r-sim) edge (1-r-simeq);
    \path (2-r-sim) edge (2-sim);
    \path (2-simeq) edge (1-r-simeq);
    \path (2-simeq) edge (2-sim);
    \path (1-r-simeq) edge (1-r-sim);
    \path (1-r-simeq) edge (1-simeq);
    \path (2-sim) edge (1-r-sim);
    \path (2-sim) edge (1-simeq);
    \path (1-r-sim) edge (1-sim);
    \path (1-simeq) edge (1-sim);
  \end{tikzpicture}
  \caption{\label{fi:spectrum}%
    The quantitative linear-time--branching-time spectrum
    from~\cite{DBLP:journals/tcs/FahrenbergL14}.  The nodes are
    different system distances, and an edge $d_1\longrightarrow d_2$
    or $d_1\dashrightarrow d_2$ indicates that
    $d_1( s, t)\ge d_2( s, t)$ for all states $s$, $t$, and that $d_1$
    and $d_2$ in general are topologically inequivalent.}
\end{figure}
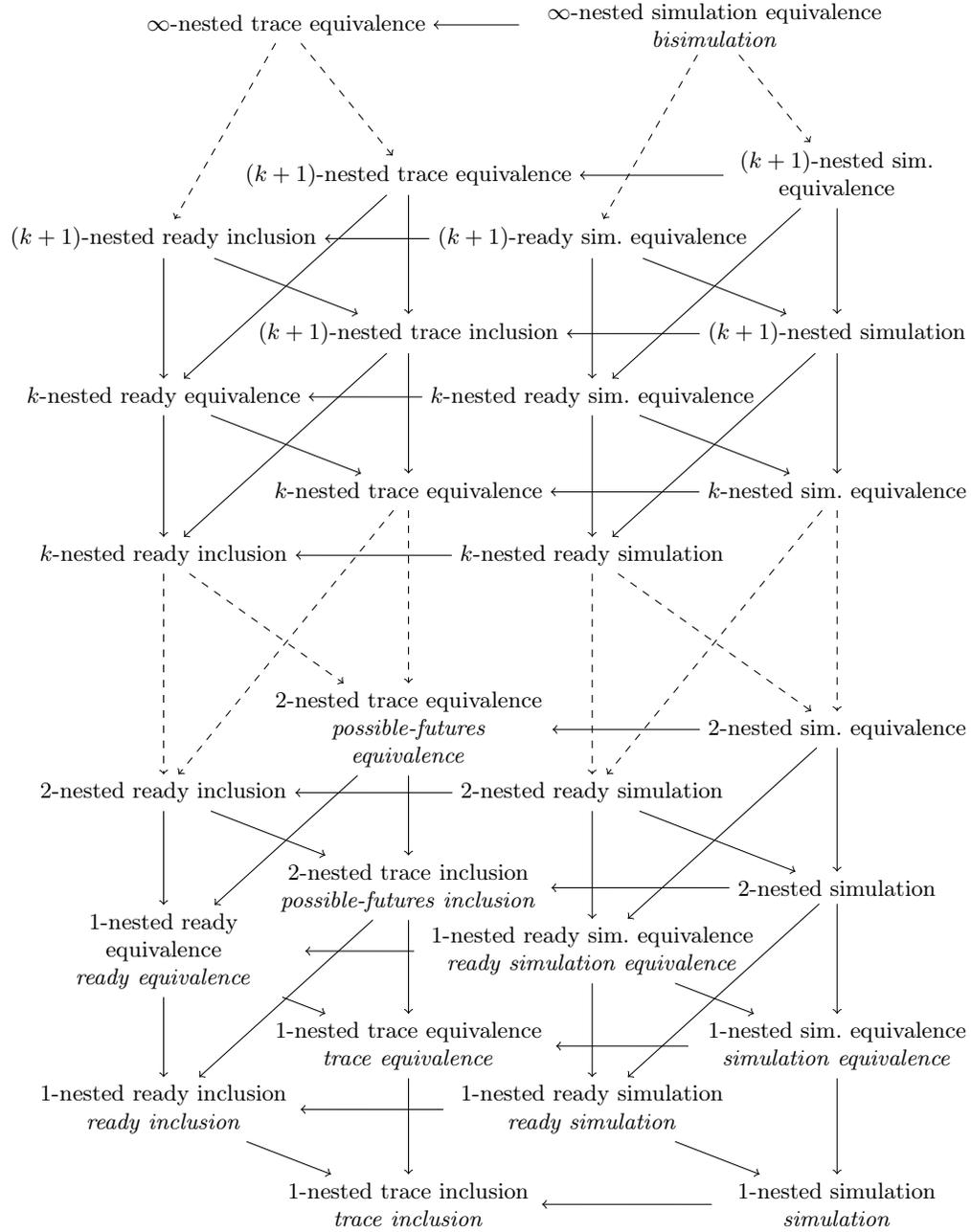

What is missing in~\cite{DBLP:journals/tcs/FahrenbergL14} are actual
\emph{algorithms} for computing the different types of distances.
(The fixed-point equations mentioned above are generally defined over
infinite lattices, hence Tarski's fixed-point theorem does not help
here.)  In this paper, we take a different route to compute them.  We
translate the general quantitative games used
in~\cite{DBLP:journals/tcs/FahrenbergL14} to other, path-building
games.  We show that under mild conditions, this translation can
always be effectuated, and that for all common trace distances, the
resulting path-building games can be solved using various methods
which we develop.

We start the paper by reviewing the quantitative games used to define
linear and branching distances
in~\cite{DBLP:journals/tcs/FahrenbergL14} in Section~\ref{se:lbdist}.
Then we show the reduction to path-building games in
Section~\ref{se:reduc} and apply this to show how to compute all
common branching distances in Section~\ref{se:comput}.  We collect our
results in the concluding section~\ref{se:conc}.  The contributions of
this paper are the following:
\begin{enumerate}[$(1)$]
\item \label{en:contri.reduce} A general method to reduce quantitative
  bisimulation-type games to path-building games.  The former can be
  posed as \emph{double} path-building games, where the players
  alternate to build \emph{two} paths; we show how to transform such
  games into a form where the players instead build \emph{one} common
  path.
\item \label{en:contri.solve} A collection of methods for solving
  different types of path-building games.  Standard methods are
  available for solving discounted games and mean-payoff games; for
  other types we develop new methods.
\item The application of the methods in~\eqref{en:contri.solve} to
  compute various types of distances between labeled transition
  systems defined by the games of~\eqref{en:contri.reduce}.
\end{enumerate}

\section{Linear and Branching Distances}
\label{se:lbdist}

Let $\Sigma$ be a set of labels.  $\Sigma^\omega$ denotes the set of
infinite traces over $\Sigma$.  We generally count sequences from
index $0$, so that $\sigma=( \sigma_0, \sigma_1,\dotsc)$.
%
Let $\Realni= \Realnn\cup\{ \infty\}$ denote the extended non-negative
real numbers.

\subsection{Trace Distances}
\label{se:trace_distances}

A \emph{trace distance} is a hemimetric $D: \Sigma^\omega\times
\Sigma^\omega\to \Realni$, \ie~a function which satisfies $D( \sigma,
\sigma)= 0$ and $D( \sigma, \tau)+ D( \tau, \upsilon)\ge D( \sigma,
\upsilon)$ for all $\sigma, \tau, \upsilon\in \Sigma^\omega$.

The following is an exhaustive list of different trace distances which
have been used in different applications.  We refer
to~\cite{DBLP:journals/tcs/FahrenbergL14} for more details and
motivation.

\paragraph{The discrete trace distance:}
$D_\textup{disc}( \sigma, \tau)= 0$ if $\sigma= \tau$ and $\infty$
otherwise.  This is equivalent to the standard Boolean setting: traces
are either equal (distance~$0$) or not (distance~$\infty$).

\paragraph{The point-wise trace distance:}
$D_\textup{sup}( \sigma, \tau)= \sup_{ n\ge 0} d( \sigma_n, \tau_n)$,
for any given label distance $d: \Sigma\times \Sigma\to \Realni$.
This measures the greatest individual symbol distance in the traces
and has been used for quantitative verification in, among others,
\cite{DBLP:journals/tse/AlfaroFS09, DBLP:conf/qest/DesharnaisLT08,
  FahrenbergLT10, LarsenFT11-Axioms, DBLP:journals/jlp/ThraneFL10,
  conf/icalp/AlfaroHM03}.

\paragraph{The discounted trace distance:}
$D_+( \sigma, \tau)= \sum_{ n= 0}^\infty \lambda^n d( \sigma_n,
\tau_n)$, for any given \emph{discounting factor}
$\lambda\in[ 0, 1\mathclose[$.  Sometimes also called
\emph{accumulating} trace distance, this accumulates individual symbol
distances along traces, using discounting to adjust the values of
distances further off.  It has been used in, for example,
\cite{FahrenbergLT10, LarsenFT11-Axioms, DBLP:journals/jlp/ThraneFL10,
  DBLP:journals/tcs/CernyHR12}.

\paragraph{The limit-average trace distance:}
$D_\textup{lavg}( \sigma, \tau)= \liminf_{ n\ge 1} \frac1 n \sum_{ i=
  0}^{ n- 1} d( \sigma_i, \tau_i)$.  This again accumulates individual
symbol distances along traces and has been used in, among others,
\cite{conf/csl/ChatterjeeDH08, DBLP:journals/tcs/CernyHR12}.  Both
discounted and limit-average distances are well-known from the theory
of discounted and mean-payoff games~\cite{EhrenfeuchtM79,
  DBLP:journals/tcs/ZwickP96}.
  
\paragraph{The Cantor trace distance:}
$D_\textup{C}( \sigma, \tau)= \frac1{ 1+ \inf\{ n\mid \sigma_n\ne
  \tau_n\}}$.  This measures the (inverse of the) length of the common
prefix of the traces and has been used for verification
in~\cite{DBLP:conf/acsd/DoyenHLN10}.

\paragraph{The maximum-lead trace distance:}
$D_\pm( \sigma, \tau)= \sup_{ n\ge0}\bigl| \sum_{ i= 0}^n( \sigma_i-
\tau_i)\bigr|$.  Here it is assumed that $\Sigma$ admits arithmetic
operations of $+$ and $-$, for instance $\Sigma\subseteq \Real$.  As
this measures differences of accumulated labels along runs, it is
especially useful for real-time systems,
\cf~\cite{DBLP:conf/formats/HenzingerMP05, conf/fit/FahrenbergL12,
  DBLP:journals/jlp/ThraneFL10}.

\subsection{Labeled Transition Systems}

A \emph{labeled transition system} (LTS) over $\Sigma$ is a tuple $(
S, i, T)$ consisting of a set of states $S$, with initial state $i\in
S$, and a set of transitions $T\subseteq S\times \Sigma\times S$.  We
often write $s\tto a t$ to mean $( s, a, t)\in T$.  We say that
$(S,i,T)$ is \emph{finite} if $S$ and $T$ are finite.  We assume our
LTS to be \emph{non-blocking} in the sense that for every state $s\in
S$ there is a transition $( s, a, t)\in T$.

We have shown in~\cite{DBLP:journals/tcs/FahrenbergL14} how any given
trace distance $D$ can be lifted to a quantitative
linear-time--branching-time spectrum of distances on LTS.  This is
done via quantitative games as we shall review below.  The point
of~\cite{DBLP:journals/tcs/FahrenbergL14} was that if the given trace
distance has a recursive formulation, which, as we show
in~\cite{DBLP:journals/tcs/FahrenbergL14}, every commonly used trace
distance has, then the corresponding linear and branching distances
can be formulated as fixed points for certain monotone functionals.

The fixed-point formulation of~\cite{DBLP:journals/tcs/FahrenbergL14}
does not, however, give rise to actual algorithms for computing linear
and branching distances, as it happens more often than not that the
mentioned monotone functionals are defined over infinite lattices.
Concretely, this is the case for all but the point-wise trace
distances in Section~\ref{se:trace_distances}.  Hence other methods
are required for computing them; developing these is the purpose of
this paper.

\subsection{Quantitative Ehrenfeucht-Fra{\"\i}ss{\'e} Games}

We review the quantitative games used
in~\cite{DBLP:journals/tcs/FahrenbergL14} to define different types of
linear and branching distances for any given trace distance $D$.  For
conciseness, we only introduce \emph{simulation games} and
\emph{bisimulation games} here, but similar definitions may be given
for all equivalences and preorders in the linear-time--branching-time
spectrum~\cite{inbook/hpa/Glabbeek01}.

\subsubsection*{Quantitative Simulation Games}

Let $\mcal S=( S, i, T)$ and $\mcal S'=( S', i', T')$ be LTS and $D:
\Sigma^\omega\times \Sigma^\omega\to \Realni$ a trace distance.  The
\emph{simulation game} from $\mcal S$ to $\mcal S'$ is played by two players,
the maximizer and the minimizer.  A play begins with the maximizer
choosing a transition $( s_0, a_0, s_1)\in T$ with $s_0= i$.  Then the
minimizer chooses a transition $( s_0', a_0', s_1')\in T'$ with $s_0'=
i'$.  Now the maximizer chooses a transition $( s_1, a_1, s_2)\in T$,
then the minimizer chooses a transition $( s_1', a_1', s_2')\in T'$,
and so on indefinitely.  Hence this is what should be called a
\emph{double path-building game}: the players each build,
independently, an infinite path in their respective LTS.

A~\emph{play} hence consists of two infinite paths, $\pi$ starting
from $i$, and $\pi'$ starting from $i'$.  The~\emph{utility} of this
play is the distance $D( \sigma, \sigma')$ between the traces
$\sigma$, $\sigma'$ of the paths $\pi$ and $\pi'$, which the maximizer
wants to maximize and the minimizer wants to minimize.  The
\emph{value} of the game is, then, the utility of the play which
results when both maximizer and minimizer are playing
optimally.

To formalize the above intuition, we define a \emph{configuration} for
the maximizer to be a pair $( \pi, \pi')$ of finite paths of equal
length, $\pi$ in $\mcal S$ and starting in $i$, $\pi'$ in $\mcal S'$ starting
in $i'$.  The intuition is that this covers the \emph{history} of a
play; the choices both players have made up to a certain point in the
game.  Hence a configuration for the minimizer is a similar pair $(
\pi, \pi')$ of finite paths, but now $\pi$ is one step longer than
$\pi'$.

A \emph{strategy} for the maximizer is a mapping from maximizer
configurations to transitions in $\mcal S$, fixing the maximizer's
choice of a move in the given configuration.  Denoting the set of
maximizer configurations by $\conf$, such a strategy is hence a
mapping $\theta: \conf\to T$ such that for all $( \pi, \pi')\in \conf$
with $\theta( \pi, \pi')=( s, a, t)$, we have $\pend( \pi)= s$.  Here
$\pend( \pi)$ denotes the last state of $\pi$.  Similarly, and
denoting the set of minimizer configurations by $\conf'$, a strategy
for the minimizer is a mapping $\theta': \conf'\to T'$ such that for
all $( \pi, \pi')\in \conf'$ with $\theta'( \pi, \pi')=( s', a', t')$,
$\pend( \pi')= s'$.

Denoting the sets of these strategies by $\Theta$ and $\Theta'$,
respectively, we can now define the \emph{simulation distance} from
$\mcal S$ to $\mcal S'$ induced by the trace distance $D$, denoted
$D^\textup{sim}( \mcal S, \mcal S')$, by
\begin{equation*}
  D^\textup{sim}( \mcal S, \mcal S')= \adjustlimits \sup_{ \theta\in \Theta}
  \inf_{ \theta'\in \Theta'} D( \sigma( \theta, \theta'), \sigma'(
  \theta, \theta'))\,,
\end{equation*}
where $\sigma( \theta, \theta')$ and $\sigma'( \theta, \theta')$ are
the traces of the paths $\pi( \theta, \theta')$ and $\pi'( \theta,
\theta')$ induced by the pair of strategies $( \theta, \theta')$.

\begin{remark}
  \label{re:simulation_game}
  If the trace distance $D$ is discrete, \ie~$D= D_\textup{disc}$ as
  in Section~\ref{se:trace_distances}, then the quantitative game
  described above reduces to the well-known \emph{simulation
    game}~\cite{DBLP:conf/banff/Stirling95}: The only choice the
  minimizer has for minimizing the value of the game is to always
  choose a transition with the same label as the one just chosen by
  the maximizer; similarly, the maximizer needs to try to force the
  game into states where she can choose a transition which the
  minimizer cannot match.  Hence the value of the game will be $0$ if
  the minimizer always can match the maximizer's labels, that is, iff
  $\mcal S$ is simulated by $\mcal S'$.
\end{remark}

\subsubsection*{Quantitative Bisimulation Games}

There is a similar game for computing the \emph{bisimulation distance}
between LTS $\mcal S$ and $\mcal S'$.  Here we give the maximizer the
choice, at each step, to either choose a transition from $s_k$ as
before, or to ``switch sides'' and choose a transition from $s_k'$
instead; the minimizer then has to answer with a transition on the
other side.

Hence the players are still building two paths, one in each LTS, but
now they are \emph{both} contributing to \emph{both} paths.  The
utility of such a play is still the distance between these two paths,
which the maximizer wants to maximize and the minimizer wants to
minimize.  The \emph{bisimulation distance} between $\mcal S$ and
$\mcal S'$, denoted $D^\textup{bisim}( \mcal S, \mcal S')$, is then
defined to be the value of this quantitative bisimulation game.


\begin{remark}
  If the trace distance $D= D_\textup{disc}$ is discrete, then using
  the same arguments as in Remark~\ref{re:simulation_game}, we see
  that $D_\textup{disc}^\textup{bisim}( \mcal S, \mcal S')= 0$ iff $\mcal S$ and
  $\mcal S'$ are \emph{bisimilar}.  The game which results being played is
  precisely the bisimulation game
  of~\cite{DBLP:conf/banff/Stirling95}, which also has been introduced
  by Fra{\"\i}ss{\'e}~\cite{Fraisse54} and
  Ehrenfeucht~\cite{Ehrenfeucht61} in other contexts.
\end{remark}

\subsubsection*{The Quantitative Linear-Time--Branching-Time Spectrum}

The above-defined quantitative simulation and bisimulation games can
be generalized using different methods.  One is to introduce a
\emph{switch counter} $\sco$ into the game which counts how often the
maximizer has switched sides during an ongoing game.  Then one can
limit the maximizer's capabilities by imposing limits on $\sco$: if
the limit is $\sco= 0$, then the players are playing a simulation
game; if there is no limit ($\sco\le \infty$), they are playing a
bisimulation game.  Other limits $\sco\le k$, for $k\in \Nat$, can be
used to define \emph{$k$-nested simulation distances}, generalizing
the equivalences and preorders
from~\cite{DBLP:journals/iandc/GrooteV92,
  DBLP:journals/jacm/HennessyM85}.

Another method of generalization is to introduce \emph{ready moves}
into the game.  These consist of the maximizer challenging her
opponent by switching sides, but only requiring that the minimizer
match the chosen transition; afterwards the game finishes.  This can
be employed to introduce the \emph{ready simulation distance}
of~\cite{DBLP:conf/popl/LarsenS89} and, combined with the switch
counter method above, the \emph{ready $k$-nested simulation distance}.
We refer to~\cite{DBLP:journals/tcs/FahrenbergL14} for further details
on these and other variants of quantitative (bi)simulation games.

For reasons of exposition, we will below introduce our reduction to
path-building games only for the quantitative simulation and
bisimulation games; but all our work can easily be transferred to the
general setting of~\cite{DBLP:journals/tcs/FahrenbergL14}.

\section{Reduction}
\label{se:reduc}

In order to compute simulation and bisimulation distances, we
translate the games of the previous section to path-building games
{\`a} la Ehrenfeucht-Mycielski \cite{EhrenfeuchtM79}.  Let
$D: \Sigma^\omega\times \Sigma^\omega\to \Realni$ be a trace distance,
and assume that there are functions
$\val_D: \Realni^\omega\to \Realni$ and
$f_D: \Sigma\times \Sigma\to \Realni$ for which it holds, for all
$\sigma, \tau\in \Sigma^\infty$, that
\begin{equation}
  \label{eq:DtoGf}
  D( \sigma, \tau)= \val_D( 0, f_D( \sigma_0, \tau_0), 0, f_D(
  \sigma_1, \tau_1), 0,\dotsc)\,.
\end{equation}
We will need these functions in our translation, and we show in
Section~\ref{se:examples-red} below that they exist for all common trace
distances.

\subsection{Simulation Distance}

Let $\mcal S=( S, i, T)$ and $\mcal S'=( S', i', T')$ be LTS.  We
construct a turn-based game
$\mcal U= \mcal U( \mcal S, \mcal S')=( U, u_0, \mato)$ as follows,
with $U= U_1\cup U_2$:
\begin{gather*}
  U_1= S\times S'\qquad U_2= S\times S'\times \Sigma\qquad u_0=( i, i')
  \\[.5ex]
  \begin{aligned}
    \mato &= \{( s, s')\ato 0( t, s', a)\mid( s, a, t)\in T\} \\[-.5ex]
    &\quad\cup \{( t, s', a)\ato{ f_D( a, a')}( t, t')\mid( s', a',
    t')\in T'\}
  \end{aligned}
\end{gather*}
This is a two-player game.  We again call the players maximizer and
minimizer, with the maximizer controlling the states in $U_1$ and the
minimizer the ones in $U_2$.  Transitions are labeled with extended
real numbers, but as the image of $f_D$ in $\Realni$ is finite, the
set of transition labels in $U$ is finite.

The game on $\mcal U$ is played as follows.  A play begins with the
maximizer choosing a transition $(u_0,a_0,u_1)\in \mato$ with $u_0=
i$.  Then the minimizer chooses a transition $(u_1,a_1,u_2)\in \mato$.
Then the maximizer chooses a transition $(u_2,a_2,u_3)\in \mato$, and
so on indefinitely (note that $\mcal U$ is non-blocking).  A play thus
induces an infinite path $\pi=(u_0,a_0,u_1), (u_1,a_1,u_2),\dotsc$ in
$\mcal U$ with $u_0= i$.  The goal of the maximizer is to maximize the
value $\val_D( \mcal U):= \val_D(a_0, a_1,\dotsc)$ of the trace of $\pi$;
the goal of the minimizer is to minimize this value.

This is hence a path-building game, variations of which (for different
valuation functions) have been studied widely in both economics and
computer science since Ehrenfeucht-Mycielski's~\cite{EhrenfeuchtM79}.
Formally, configurations and strategies are given as follows.  A
configuration of the maximizer is a path $\pi_1$ in $\mcal U$ with
$\pend( \pi_1)\in U_1$, and a configuration of the minimizer is a path
$\pi_2$ in $\mcal U$ with $\pend( \pi_2)\in U_2$.  Denote the sets of
these configurations by $\conf_1$ and $\conf_2$, respectively.  A
strategy for the maximizer is, then, a mapping
$\theta_1: \conf_1\to \mato$ such that for all $\pi_1\in \conf_1$ with
$\theta_1( \pi_1)=( u, x, v)$, $\pend( \pi_1)= u$.  A strategy for the
minimizer is a mapping $\theta_2: \conf_2\to \mato$ such that for all
$\pi_2\in \conf_2$ with $\theta_2( \pi_2)=( u, x, v)$,
$\pend( \pi_2)= u$.  Denoting the sets of these strategies by
$\Theta_1$ and $\Theta_2$, respectively, we can now define
\begin{equation*}
  \val_D( \mcal U)= \adjustlimits \sup_{ \theta_1\in \Theta_1}
  \inf_{ \theta_2\in \Theta_2} \val_D( \sigma( \theta_1,
  \theta_2))\,,
\end{equation*}
where $\sigma( \theta_1, \theta_2)$ is the trace of the path $\pi(
\theta_1, \theta_2)$ induced by the pair of strategies $( \theta_1,
\theta_2)$.

By the next theorem, the value of $\mcal U$ is precisely the simulation
distance from $\mcal S$ to $\mcal S'$.

\begin{theorem}
  \label{th:sim=u}
  For all LTS $\mcal S$, $\mcal S'$, $D^\textup{sim}( \mcal S, \mcal S')= \val_D( \mcal U(
  \mcal S, \mcal S'))$.
\end{theorem}

\begin{proof}
  Write $\mcal S=( S, i, T)$ and $\mcal S'=( S', i', T')$.  Informally, the
  reason for the equality is that any move $( s, a, t)\in T$ of the
  maximizer in the simulation distance game can be copied to a move $(
  s, s')\ato{ 0}( t, s', a)$, regardless of $s'$, in $\mcal U$.
  Similarly, any move $( s', a', t')$ of the minimizer can be copied
  to a move $( t, s', a)\ato{ f_D( a, a')}( t, t')$, and all the moves
  in $\mcal U$ are of this form.

  To turn this idea into a formal proof, we show that there are
  bijections between configurations and strategies in the two games,
  and that under these bijections, the utilities of the two games are
  equal.  For $( \pi, \pi')\in \conf$ in the simulation distance game,
  with $\pi=( s_0, a_0, s_1),\dotsc,$ $( s_{ n- 1}, a_{ n- 1}, s_n)$
  and $\pi'=( s_0', a_0', s_1'),\dotsc,( s_{ n- 1}', a_{ n- 1}',
  s_n')$, define
  \begin{multline*}
    \phi_1( \pi, \pi')=(( s_0, s_0'), 0,( s_1, s_0', a_0)),(( s_1,
    s_0', a_0), f_D( a_0, a_0'),( s_1, s_1')),\dotsc, \\
    (( s_n, s_{ n- 1}', a_{ n- 1}), f_D( a_{ n- 1}, a_{ n- 1}'),( s_n,
    s_n'))\,.
  \end{multline*}
  It is clear that this defines a bijection $\phi_1: \conf\to
  \conf_1$, and that one can similarly define a bijection $\phi_2:
  \conf'\to \conf_2$.

  Now for every strategy $\theta: \conf\to T$ in the simulation
  distance game, define a strategy $\psi_1( \theta)= \theta_1\in
  \Theta_1$ as follows.  For $\pi_1\in \conf_1$, let $( \pi, \pi')=
  \phi_1^{ -1}( \pi_1)$ and $s'= \pend( \pi')$.  Let $\theta( \pi,
  \pi')=( s, a, t)$ and define $\theta_1( \pi_1)=(( s, s'), 0,( t, s',
  a))$.  Similarly we define a mapping $\psi_2: \Theta'\to \Theta_2$ as
  follows.  For $\theta': \conf'\to T'$ and $\pi_2\in \conf_2$, let $(
  \pi, \pi')= \phi_2^{ -1}( \pi_2)$ with $\pi=( s_0, a_0,
  s_1),\dotsc,( s_n, a_n, s_{ n+ 1})$.  Let $\theta'( \pi, \pi')=( s',
  a', t')$ and define $\psi_2( \theta')( \pi_2)=((s_{ n+ 1}, s', a_n),
  f_D( a_n, a'),( s_{ n+ 1}, t'))$.

  It is clear that $\psi_1$ and $\psi_2$ indeed map strategies in the
  simulation distance game to strategies in $\mcal U$ and that both are
  bijections.  Also, for each pair $( \theta, \theta')\in \Theta\times
  \Theta'$, $D( \sigma( \theta, \theta'), \sigma'( \theta, \theta'))=
  \val_D( \sigma( \psi_1( \theta), \psi_2( \theta')))$ by
  construction.  But then
  \begin{align*}
    D^\textup{sim}( \mcal S, \mcal S') &= \adjustlimits \sup_{ \theta\in
      \Theta} \inf_{ \theta'\in \Theta'} D( \sigma( \theta, \theta'),
    \sigma'( \theta, \theta')) \\
    &= \adjustlimits \sup_{ \theta\in \Theta} \inf_{ \theta'\in
      \Theta'} \val_D( \sigma( \psi_1( \theta), \psi_2( \theta'))) \\
    &= \adjustlimits \sup_{ \theta_1\in \Theta_1} \inf_{ \theta_2\in
      \Theta_2} \val_D( \sigma( \theta_1, \theta_2))= \val_D( \mcal U)\,,
  \end{align*}
  the third equality because $\psi_1$ and $\psi_2$ are bijections. \qed
\end{proof}

\subsection{Examples}
\label{se:examples-red}

  We show that the reduction applies to all trace distances from
  Section~\ref{se:trace_distances}.
  \begin{enumerate}
  \item \label{ex:red.discr} For the discrete trace distance $D=
    D_\textup{disc}$, we let
    \begin{equation*}
      \val_D( x)= \sum_{ n= 0}^\infty x_n\,, \qquad f_D( a, b)=
      \begin{cases}
        0 &\text{if } a= b\,, \\
        \infty &\text{otherwise}\,,
      \end{cases}
    \end{equation*}
    then~\eqref{eq:DtoGf} holds.  In the game on $\mcal U$, the minimizer
    needs to play $0$-labeled transitions to keep the distance at $0$.
  \item \label{ex:red.pw} For the point-wise trace distance $D=
    D_\textup{sup}$, we can let
    \begin{equation*}
      \val_D( x)= \sup_{ n\ge 0} x_n\,, \qquad f_D( a, b)= d( a, b)\,.
    \end{equation*}
    Hence the game on $\mcal U$ computes the sup of a trace.
  \item \label{ex:red.disc} For the discounted trace distance $D=
    D_+$, let
    \begin{equation*}
      \val_D( x)= \sum_{ n= 0}^\infty \sqrt \lambda^n x_n\,, \qquad
      f_D( a, b)= \sqrt \lambda\, d( a, b)\,,
    \end{equation*}
    then~\eqref{eq:DtoGf} holds.  Hence the game on $\mcal U$ is a
    standard discounted game~\cite{DBLP:journals/tcs/ZwickP96}.
  \item \label{ex:red.lavg} For the limit-average trace distance $D=
    D_\textup{lavg}$, we can let
    \begin{equation*}
      \val_D( x)= \liminf_{ n\ge 1} \frac
      1 n \sum_{ i= 0}^{ n- 1} x_i\,, \qquad f_D( a, b)= 2 d( a, b)\,;
    \end{equation*}
    we will show below that~\eqref{eq:DtoGf} holds.  Hence the game on
    $\mcal U$ is a mean-payoff game~\cite{DBLP:journals/tcs/ZwickP96}.
  \item \label{ex:red.cant} For the Cantor trace distance $D=
    D_\textup{C}$, let
    \begin{equation*}
      \val_D( x)= \frac2{ 1+ \inf\{ n\mid x_n\ne
        0\}}\,, \qquad f_D( a, b)=
      \begin{cases}
        0 &\text{if } a= b\,, \\
        1 &\text{otherwise}\,.
      \end{cases}
    \end{equation*}
    The objective of the maximizer in this game is to reach a
    transition with weight $1$ \emph{as soon as possible}.
  \item \label{ex:red.maxl} For the maximum-lead trace distance $D=
    D_\pm$, we can let
    \begin{equation*}
      \val_D( x)= \sup_{ n\ge 0}\bigl| \sum_{ i=
      0}^n x_i\bigr|\,, \qquad f_D( a, b)= a- b\,,
    \end{equation*}
    then~\eqref{eq:DtoGf} holds.
  \end{enumerate}

\subsection{Bisimulation Distance}

We can construct a similar turn-based game to compute the bisimulation
distance.  Let $\mcal S=( S, i, T)$ and $\mcal S'=( S', i', T')$ be
LTS and define $\mcal V= \mcal V( \mcal S, \mcal S')=( V, v_0, \mato)$
as follows, with $V= V_1\cup V_2$:
\begin{gather*}
  V_1= S\times S'\qquad V_2= S\times S'\times \Sigma\times\{ 1,
  2\}\qquad v_0=( i, i') \\[.5ex]
  \begin{aligned}
    \mato &= \{( s, s')\ato 0( t, s', a, 1)\mid( s, a, t)\in T\}
    \\[-.5ex]
    &\quad\cup \{( s, s')\ato 0( s, t', a', 2)\mid( s', a', t')\in
    T'\} \\[-.5ex]
    &\quad\cup \{( t, s', a, 1)\ato{ f_D( a, a')}( t, t')\mid( s', a',
    t')\in T'\} \\[-.5ex]
    &\quad\cup \{( s, t', a', 2)\ato{ f_D( a, a')}( t, t')\mid( s, a,
    t)\in T\}
  \end{aligned}
\end{gather*}
Here we have used the minimizer's states to both remember the label
choice of the maximizer and which side of the bisimulation game she
plays on.  By suitable modifications, we can construct similar games
for all distances in the spectrum
of~\cite{DBLP:journals/tcs/FahrenbergL14}.  The next theorem states
that the value of $\mcal V$ is precisely the bisimulation distance between
$\mcal S$ and $\mcal S'$.

\begin{theorem}
  For all LTS $\mcal S$, $\mcal S'$, $D^\textup{bisim}( \mcal S, \mcal S')= \val_D(
  \mcal V( \mcal S, \mcal S'))$.
\end{theorem}

\begin{proof}
  This proof is similar to the one of Theorem~\ref{th:sim=u}, only
  that now, we have to take into account that the maximizer may
  ``switch sides''.  The intuition is that maximizer moves $( s, a,
  t)$ in the $\mcal S$ component of the bisimulation distance games are
  emulated by moves $( s, s')\ato 0( t, s', a, 1)$, maximizer moves $(
  s', a', t')$ in the $\mcal S'$ component are emulated by moves $( s,
  s')\ato 0( s, t', a', 2)$, and similarly for the minimizer.  The
  values $1$ and $2$ in the last component of the $V_2$ states ensure
  that the minimizer only has moves available which correspond to
  playing in the correct component in the bisimulation distance game
  (\ie~that $\psi_2$ is a bijection). \qed
\end{proof}

\section{Computing the Values of Path-Building Games}
\label{se:comput}

We show here how to compute the values of the different path-building
games which we saw in the last section.  This will give us algorithms
to compute all simulation and bisimulation distances associated with
the trace distances of Section~\ref{se:trace_distances}.

We will generally only refer to the games $\mcal U$ for computing
simulation distance here, but the bisimulation distance games $\mcal V$
are very similar, and everything we say also applies to them.

\paragraph{Discrete distance:}

The game to compute the discrete simulation distances is a
reachability game, in that the goal of the maximizer is to force the
minimizer into a state from which she can only choose $\infty$-labeled
transitions.  We can hence solve them using the standard
controllable-predecessor operator defined, for any set $S\subseteq
U_1$ of maximizer states, by
\begin{equation*}
  \cpre( S)=\{ u_1\in U_1\mid \exists u_1\ato 0 u_2: \forall u_2\ato x
  u_3: u_3\in S\}\,.
\end{equation*}

Now let $S\subseteq U_1$ be the set of states from which the maximizer
can force the game into a state from which the minimizer only has
$\infty$-labeled transitions, \ie
\begin{equation*}
  S=\{ u_1\in U_1\mid \exists u_1\ato
  0 u_2: \forall u_2\ato x u_3: x= \infty\}\,,
\end{equation*}
and compute $S^*= \cpre^*( S)= \bigcup_{ n\ge 0} \cpre^n( S)$.  By
monotonicity of $\cpre$ and as the subset lattice of $U_1$ is complete
and finite, this computation finishes in at most $| U_1|$ steps.

\begin{lemma}
  $\val_D( \mcal U)= 0$ iff $u_0\notin S^*$.
\end{lemma}

\begin{proof}
  As we are working with the discrete distance, we have either
  $\val_D( \mcal U)= 0$ or $\val_D( \mcal U)= \infty$.  Now
  $u_o\in S^*$ iff the maximizer can force, using finitely many steps,
  the game into a state from which the minimizer only has
  $\infty$-labeled transitions, which is the same as $\val_D( \mcal
  U)= \infty$. \qed
\end{proof}

\paragraph{Point-wise distance:}

To compute the value of the point-wise simulation distance game, let
$W=\{ w_1,\dotsc, w_m\}$ be the (finite) set of weights of the
minimizer's transitions, ordered such that $w_1<\dotsm< w_m$.  For
each $i= 1,\dotsc, m$, let $S_i=\{ u_1\in U_1: \exists u_1\ato 0 u_2:
\forall u_2\ato x u_3: x\ge w_i\}$ be the set of maximizer states from
which the maximizer can force the minimizer into a transition with
weight at least $w_i$; note that $S_m\subseteq S_{ m-
  1}\subseteq\dotsm\subseteq S_1= U_1$.  For each $i= 1,\dotsc, m$,
compute $S_i^*= \cpre^*( S_i)$, then $S_m^*\subseteq S_{ m-
  1}^*\subseteq\dotsm\subseteq S_1^*= U_1$.

\begin{lemma}
  Let $p$ be the greatest index for which $u_0\in S_p^*$, then
  $p= \val_D( \mcal U)$.
\end{lemma}

\begin{proof}
  For any $k$, we have $u_0\in S_k^*$ iff the maximizer can force,
  using finitely many steps, the game into a state from which the
  minimizer only has transitions with weight at least $w_k$.  Thus
  $u_0\in S_p^*$ iff (1)~the maximizer can force the minimizer into a
  $w_p$-weighted transition; (2)~the maximizer \emph{cannot} force the
  minimizer into a $w_{ p+ 1}$-weighted transition. \qed
\end{proof}

\paragraph{Discounted distance:}

The game to compute the discounted simulation distance is a standard
discounted game and can be solved by standard
methods~\cite{DBLP:journals/tcs/ZwickP96}.

\paragraph{Limit-average distance:}

For the limit-average simulation distance game, let $( y_n)_{ n\ge 1}$
be the sequence $( 1, 1, \frac32, 1, \frac54,\dotsc)$ and note that
$\lim_{ n\to \infty} y_n= 1$.  Then
\begin{align*}
  \val_D( x)= \val_D( x) \lim_{ n\to \infty} y_n &= \liminf_{ n\ge 1}
  \frac{ y_n}{ n} \sum_{ i= 0}^{ n- 1} x_i \\
  &= \liminf_{ 2k\ge 1} \frac 1{ 2k} \sum_{ i= 0}^{ k- 1} f_D(
  \sigma_i, \tau_i) \\
  &= \liminf_{ k\ge 1} \frac 1 k \sum_{ i= 0}^{ k- 1} d( \sigma_i,
  \tau_i)= D_\textup{lavg}( \sigma, \tau)\,,
\end{align*}
so, indeed, \eqref{eq:DtoGf}~holds.  The game is a standard
mean-payoff game and can be solved by standard methods, see for
example~\cite{DBLP:conf/valuetools/DhingraG06}.

\paragraph{Cantor distance:}

To compute the value of the Cantor simulation distance game, let
$S_1\subseteq U_1$ be the set of states from which the maximizer can
force the game into a state from which the minimizer only has
$1$-labeled transitions, \ie~$S_1=\{ u_1\in U_1\mid \exists u_1\tto 0
u_2: \forall u_2\tto x u_3: x= 1\}$.  Now recursively compute $S_{ i+
  1}= S_i\cup \cpre( S_i)$, for $i= 1, 2,\dotsc$, until $S_{ i+ 1}=
S_i$ (which, as $S_i\subseteq S_{ i+ 1}$ for all $i$ and $U_1$ is
finite, will happen eventually).  Then $S_i$ is the set of states from
which the maximizer can force the game to a $1$-labeled minimizer
transition which is at most $2i$ steps away.  Hence $\val_D( \mcal U)= 0$
if there is no $p$ for which $u_0\in S_p$, and otherwise $\val_D(
\mcal U)= \frac1 p$, where $p$ is the least index for which $u_0\in S_p$.

\paragraph{Maximum-lead distance:}

For the maximum-lead simulation distance game, we note that the
maximizer wants to maximize
$\sup_{ n\ge 0}\bigl| \sum_{ i= 0}^n x_i\bigr|$, \ie~wants the
accumulated values $\sum_{ i= 0}^n x_i$ or $-\sum_{ i= 0}^n x_i$ to
exceed any prescribed bounds.  A weighted game in which one player
wants to keep accumulated values inside some given bounds, while the
opponent wants to exceed these bounds, is called an
\emph{interval-bound energy game}.  It is shown
in~\cite{DBLP:conf/formats/BouyerFLMS08} that solving general
interval-bound energy games is EXPTIME-complete.

We can reduce the problem of computing maximum-lead simulation
distance to an interval-bound energy game by first
non-deterministically choosing a bound $k$ and then checking whether
player~1 wins the interval-bound energy game on $\mcal U$ for bounds
$[ -k, k]$.  (There is a slight problem in that
in~\cite{DBLP:conf/formats/BouyerFLMS08}, energy games are defined
only for \emph{integer}-weighted transition systems, whereas we are
dealing with real weights here.  However, it is easily seen that the
results of~\cite{DBLP:conf/formats/BouyerFLMS08} also apply to
\emph{rational} weights and bounds; and as our transition systems are
finite, one can always find a sound and complete rational
approximation.)

We can thus compute maximum-lead simulation distance in
non-deterministic exponential time; we leave open for now the question
whether there is a more efficient algorithm.


\section{Conclusion and Future Work}
\label{se:conc}

We sum up our results in the following corollary which gives the
complexities of the decision problems associated with the respective
distance computations.  Note that the first part restates the
well-known fact that simulation and bisimulation are decidable in
polynomial time.

\begin{corollary}
  \mbox{}
  \begin{enumerate}
  \item Discrete simulation and bisimulation distances are computable
    in PTIME.
  \item Point-wise simulation and bisimulation distances are
    computable in PTIME.
  \item Discounted simulation and bisimulation distances are
    computable in $\text{NP}\cap\text{coNP}$.
  \item Limit-average simulation and bisimulation distances are
    computable in $\text{NP}\cap\text{coNP}$.
  \item Cantor simulation and bisimulation distances are computable in
    PTIME.
  \item Maximum-lead simulation and bisimulation distances are
    computable in \linebreak NEXPTIME.
  \end{enumerate}
\end{corollary}

In the future, we intend to expand our work to also cover
\emph{quantitative specification theories}.  Together with several
coauthors, we have in~\cite{DBLP:conf/facs2/FahrenbergKLT14,
  DBLP:journals/soco/FahrenbergKLT18} developed a comprehensive
setting for satisfaction and refinement distances in quantitative
specification theories.  Using our work
in~\cite{DBLP:conf/sofsem/FahrenbergL17} on a qualitative
linear-time--branching-time spectrum of specification theories, we
plan to introduce a quantitative linear-time--branching-time spectrum
of specification distances and to use the setting developed here to
devise methods for computing them through path-building games.

Another possible extension of our work contains \emph{probabilistic}
systems, for example the probabilistic automata
of~\cite{DBLP:conf/concur/SegalaL94}.  A possible starting point for
this is~\cite{DBLP:conf/birthday/BreugelW14} which uses simple
stochastic games to compute probabilistic bisimilarity.

\bibliographystyle{plain}

\begin{thebibliography}{10}

\bibitem{DBLP:journals/tcs/AlurCHHHNOSY95}
Rajeev Alur, Costas Courcoubetis, Nicolas Halbwachs, Thomas~A. Henzinger,
  Pei-Hsin Ho, Xavier Nicollin, Alfredo Olivero, Joseph Sifakis, and Sergio
  Yovine.
\newblock The algorithmic analysis of hybrid systems.
\newblock {\em {Theor. Comput. Sci.}}, 138(1):3--34, 1995.

\bibitem{DBLP:journals/tcs/AlurD94}
Rajeev Alur and David~L. Dill.
\newblock A theory of timed automata.
\newblock {\em {Theor. Comput. Sci.}}, 126(2):183--235, 1994.

\bibitem{DBLP:journals/tocl/AzizSSB00}
Adnan Aziz, Kumud Sanwal, Vigyan Singhal, and Robert~K. Brayton.
\newblock Model-checking continous-time {M}arkov chains.
\newblock {\em ACM Transactions on Computational Logics}, 1(1):162--170, 2000.

\bibitem{DBLP:conf/formats/BouyerFLMS08}
Patricia Bouyer, Uli Fahrenberg, Kim~Guldstrand Larsen, Nicolas Markey, and
  Ji{\v r}{\'{\i}} Srba.
\newblock Infinite runs in weighted timed automata with energy constraints.
\newblock In Franck Cassez and Claude Jard, editors, {\em FORMATS}, volume 5215
  of {\em {Lect. Notes Comput. Sci.}}, pages 33--47. {Springer-Verlag}, 2008.

\bibitem{DBLP:conf/atva/BygJS09}
Joakim Byg, Kenneth~Yrke J{\o}rgensen, and Ji{\v r}{\'{\i}} Srba.
\newblock {TAPAAL:} editor, simulator and verifier of timed-arc {P}etri nets.
\newblock In Zhiming Liu and Anders~P. Ravn, editors, {\em ATVA}, volume 5799
  of {\em {Lect. Notes Comput. Sci.}}, pages 84--89. {Springer-Verlag}, 2009.

\bibitem{DBLP:journals/tcs/CernyHR12}
Pavol {\v C}ern{\'y}, Thomas~A. Henzinger, and Arjun Radhakrishna.
\newblock Simulation distances.
\newblock {\em {Theor. Comput. Sci.}}, 413(1):21--35, 2012.

\bibitem{conf/csl/ChatterjeeDH08}
Krishnendu Chatterjee, Laurent Doyen, and Thomas~A. Henzinger.
\newblock Quantitative languages.
\newblock {\em ACM Transactions on Computational Logic}, 11(4), 2010.

\bibitem{DBLP:journals/tcs/AlfaroFHMS05}
Luca de~Alfaro, Marco Faella, Thomas~A. Henzinger, Rupak Majumdar, and
  Mari{\-}{\"e}lle Stoelinga.
\newblock Model checking discounted temporal properties.
\newblock {\em {Theor. Comput. Sci.}}, 345(1):139--170, 2005.

\bibitem{DBLP:journals/tse/AlfaroFS09}
Luca de~Alfaro, Marco Faella, and Mari{\"e}lle Stoelinga.
\newblock Linear and branching system metrics.
\newblock {\em IEEE Transactions on Software Engineering}, 35(2):258--273,
  2009.

\bibitem{conf/icalp/AlfaroHM03}
Luca de~Alfaro, Thomas~A. Henzinger, and Rupak Majumdar.
\newblock Discounting the future in systems theory.
\newblock In Jos C.~M. Baeten, Jan~Karel Lenstra, Joachim Parrow, and
  Gerhard~J. Woeginger, editors, {\em Automata, Languages and Programming, 30th
  International Colloquium, {ICALP} 2003}, volume 2719 of {\em Lecture Notes in
  Computer Science}, pages 1022--1037. Springer, 2003.

\bibitem{DBLP:journals/tcs/DesharnaisGJP04}
Josee Desharnais, Vineet Gupta, Radha Jagadeesan, and Prakash Panangaden.
\newblock Metrics for labelled {M}arkov processes.
\newblock {\em {Theor. Comput. Sci.}}, 318(3):323--354, 2004.

\bibitem{DBLP:conf/qest/DesharnaisLT08}
Jos{\'e}e Desharnais, Fran\c{c}ois Laviolette, and Mathieu Tracol.
\newblock Approximate analysis of probabilistic processes.
\newblock In {\em QEST}, pages 264--273. IEEE Computer Society, 2008.

\bibitem{DBLP:conf/valuetools/DhingraG06}
Vishesh Dhingra and Stephane Gaubert.
\newblock How to solve large scale deterministic games with mean payoff by
  policy iteration.
\newblock In Luciano Lenzini and Rene~L. Cruz, editors, {\em VALUETOOLS},
  volume 180 of {\em {ACM} Int. Conf. Proc.}, page~12. {ACM}, 2006.

\bibitem{DBLP:conf/acsd/DoyenHLN10}
Laurent Doyen, Thomas~A. Henzinger, Axel Legay, and Dejan Ni{\v c}kovi{\'c}.
\newblock Robustness of sequential circuits.
\newblock In Lu\'{\i}s Gomes, Victor Khomenko, and Jo{\~a}o~M. Fernandes,
  editors, {\em ACSD}, pages 77--84. IEEE Computer Society, 2010.

\bibitem{Ehrenfeucht61}
Andrzej Ehrenfeucht.
\newblock An application of games to the completeness problem for formalized
  theories.
\newblock {\em Fund. Math.}, 49:129--141, 1961.

\bibitem{EhrenfeuchtM79}
Andrzej Ehrenfeucht and Jan Mycielski.
\newblock Positional strategies for mean payoff games.
\newblock {\em Int. J. Game Theory}, 8:109--113, 1979.

\bibitem{DBLP:conf/facs2/FahrenbergKLT14}
Uli Fahrenberg, Jan K{\v r}et{\'{\i}}nsk{\'{y}}, Axel Legay, and Louis{-}Marie
  Traonouez.
\newblock Compositionality for quantitative specifications.
\newblock In Ivan Lanese and Eric Madelaine, editors, {\em FACS}, volume 8997
  of {\em {Lect. Notes Comput. Sci.}}, pages 306--324. {Springer-Verlag}, 2014.

\bibitem{DBLP:journals/soco/FahrenbergKLT18}
Uli Fahrenberg, Jan K{\v r}et{\'{\i}}nsk{\'{y}}, Axel Legay, and Louis{-}Marie
  Traonouez.
\newblock Compositionality for quantitative specifications.
\newblock {\em {Soft Comput.}}, 22(4):1139--1158, 2018.

\bibitem{FahrenbergLT10}
Uli Fahrenberg, Kim~G. Larsen, and Claus Thrane.
\newblock A quantitative characterization of weighted {K}ripke structures in
  temporal logic.
\newblock {\em Computing and Informatics}, 29(6+):1311--1324, 2010.

\bibitem{conf/fit/FahrenbergL12}
Uli Fahrenberg and Axel Legay.
\newblock A robust specification theory for modal event-clock automata.
\newblock In Sebastian~S. Bauer and Jean{-}Baptiste Raclet, editors, {\em
  Proceedings Fourth Workshop on Foundations of Interface Technologies, {FIT}
  2012, Tallinn, Estonia, 25th March 2012.}, volume~87 of {\em {EPTCS}}, pages
  5--16, 2012.

\bibitem{DBLP:journals/tcs/FahrenbergL14}
Uli Fahrenberg and Axel Legay.
\newblock The quantitative linear-time-branching-time spectrum.
\newblock {\em {Theor. Comput. Sci.}}, 538:54--69, 2014.

\bibitem{DBLP:conf/sofsem/FahrenbergL17}
Uli Fahrenberg and Axel Legay.
\newblock A linear-time-branching-time spectrum of behavioral specification
  theories.
\newblock In Bernhard Steffen, Christel Baier, Mark van~den Brand, Johann Eder,
  Mike Hinchey, and Tiziana Margaria, editors, {\em SOFSEM}, volume 10139 of
  {\em {Lect. Notes Comput. Sci.}}, pages 49--61. {Springer-Verlag}, 2017.

\bibitem{Fraisse54}
Roland Fra{\"i}ss{\'e}.
\newblock Sur quelques classifications des syst{\`e}mes de relations.
\newblock {\em Publ. Scient. de l'Univ. d'Alger, S{\'e}rie A}, 1:35--182, 1954.

\bibitem{DBLP:journals/fmsd/FranzleH07}
Martin Fr{\"{a}}nzle and Christian Herde.
\newblock {HySAT}: An efficient proof engine for bounded model checking of
  hybrid systems.
\newblock {\em {Formal Meth. Syst. Design}}, 30(3):179--198, 2007.

\bibitem{DBLP:conf/cav/FrehseGDCRLRGDM11}
Goran Frehse, Colas~Le Guernic, Alexandre Donz{\'e}, Scott Cotton, Rajarshi
  Ray, Olivier Lebeltel, Rodolfo Ripado, Antoine Girard, Thao Dang, and Oded
  Maler.
\newblock {SpaceEx}: Scalable verification of hybrid systems.
\newblock In Ganesh Gopalakrishnan and Shaz Qadeer, editors, {\em CAV}, volume
  6806 of {\em {Lect. Notes Comput. Sci.}}, pages 379--395. {Springer-Verlag},
  2011.

\bibitem{DBLP:conf/cav/GardeyLMR05}
Guillaume Gardey, Didier Lime, Morgan Magnin, and Olivier~H. Roux.
\newblock {Romeo}: A tool for analyzing time {P}etri nets.
\newblock In Kousha Etessami and Sriram~K. Rajamani, editors, {\em CAV}, volume
  3576 of {\em {Lect. Notes Comput. Sci.}}, pages 418--423. {Springer-Verlag},
  2005.

\bibitem{DBLP:conf/cpe/GilmoreH94}
Stephen Gilmore and Jane Hillston.
\newblock The {PEPA} workbench: A tool to support a process algebra-based
  approach to performance modelling.
\newblock In G{\"u}nter Haring and Gabriele Kotsis, editors, {\em CPE}, volume
  794 of {\em {Lect. Notes Comput. Sci.}}, pages 353--368. {Springer-Verlag},
  1994.

\bibitem{DBLP:journals/tac/GirardP07}
Antoine Girard and George~J. Pappas.
\newblock Approximation metrics for discrete and continuous systems.
\newblock {\em IEEE Trans. Automat. Contr.}, 52(5):782--798, 2007.

\bibitem{DBLP:journals/iandc/GrooteV92}
Jan~Friso Groote and Frits~W. Vaandrager.
\newblock Structured operational semantics and bisimulation as a congruence.
\newblock {\em {Inf. Comput.}}, 100(2):202--260, 1992.

\bibitem{DBLP:conf/apn/Hanisch93}
Hans-Michael Hanisch.
\newblock Analysis of place/transition nets with timed arcs and its application
  to batch process control.
\newblock In Marco~Ajmone Marsan, editor, {\em ATPN}, volume 691 of {\em {Lect.
  Notes Comput. Sci.}}, pages 282--299. {Springer-Verlag}, 1993.

\bibitem{DBLP:journals/fac/HanssonJ94}
Hans Hansson and Bengt Jonsson.
\newblock A logic for reasoning about time and reliability.
\newblock {\em Formal Aspects of Computing}, 6(5):512--535, 1994.

\bibitem{DBLP:journals/jacm/HennessyM85}
Matthew Hennessy and Robin Milner.
\newblock Algebraic laws for nondeterminism and concurrency.
\newblock {\em Journal of the ACM}, 32(1):137--161, 1985.

\bibitem{DBLP:journals/sttt/HenzingerHW97}
Thomas~A. Henzinger, Pei-Hsin Ho, and Howard Wong-Toi.
\newblock {HYTECH}: A model checker for hybrid systems.
\newblock {\em {Int. J. Softw. Tools Techn. Trans.}}, 1(1-2):110--122, 1997.

\bibitem{DBLP:conf/formats/HenzingerMP05}
Thomas~A. Henzinger, Rupak Majumdar, and Vinayak~S. Prabhu.
\newblock Quantifying similarities between timed systems.
\newblock In Paul Pettersson and Wang Yi, editors, {\em Formal Modeling and
  Analysis of Timed Systems, Third International Conference, {FORMATS} 2005,
  Uppsala, Sweden, September 26-28, 2005, Proceedings}, volume 3829 of {\em
  Lecture Notes in Computer Science}, pages 226--241. Springer, 2005.

\bibitem{DBLP:journals/iandc/HenzingerNSY94}
Thomas~A. Henzinger, Xavier Nicollin, Joseph Sifakis, and Sergio Yovine.
\newblock Symbolic model checking for real-time systems.
\newblock {\em {Inf. Comput.}}, 111(2):193--244, 1994.

\bibitem{book/Hillston96}
Jane Hillston.
\newblock {\em A Compositional Approach to Performance Modelling}.
\newblock Cambridge University Press, 1996.

\bibitem{DBLP:journals/rts/Koymans90}
Ron Koymans.
\newblock Specifying real-time properties with metric temporal logic.
\newblock {\em Real-Time Systems}, 2(4):255--299, 1990.

\bibitem{DBLP:conf/tacas/KwiatkowskaNP02}
Marta~Z. Kwiatkowska, Gethin Norman, and David Parker.
\newblock Probabilistic symbolic model checking with {PRISM}: A hybrid
  approach.
\newblock In Joost-Pieter Katoen and Perdita Stevens, editors, {\em TACAS},
  volume 2280 of {\em {Lect. Notes Comput. Sci.}}, pages 52--66.
  {Springer-Verlag}, 2002.

\bibitem{LarsenFT11-Axioms}
Kim~G. Larsen, Uli Fahrenberg, and Claus Thrane.
\newblock Metrics for weighted transition systems: {A}xiomatization and
  complexity.
\newblock {\em {Theor. Comput. Sci.}}, 412(28):3358--3369, 2011.

\bibitem{DBLP:journals/sttt/LarsenPY97}
Kim~G. Larsen, Paul Pettersson, and Wang Yi.
\newblock {UPPAAL} in a nutshell.
\newblock {\em {Int. J. Softw. Tools Techn. Trans.}}, 1(1-2):134--152, 1997.

\bibitem{DBLP:conf/popl/LarsenS89}
Kim~G. Larsen and Arne Skou.
\newblock Bisimulation through probabilistic testing.
\newblock In {\em POPL}, pages 344--352. {ACM} Press, 1989.

\bibitem{journal/transcom/MerlinF76}
Philip~M. Merlin and David~J. Farber.
\newblock Recoverability of communication protocols--implications of a
  theoretical study.
\newblock {\em {IEEE Trans. Commun.}}, 24(9):1036 -- 1043, 1976.

\bibitem{DBLP:conf/concur/SegalaL94}
Roberto Segala and Nancy~A. Lynch.
\newblock Probabilistic simulations for probabilistic processes.
\newblock In Bengt Jonsson and Joachim Parrow, editors, {\em CONCUR}, volume
  836 of {\em {Lect. Notes Comput. Sci.}}, pages 481--496. {Springer-Verlag},
  1994.

\bibitem{book/Stewart94}
William~J. Stewart.
\newblock {\em Introduction to the Numerical Solution of {M}arkov Chains}.
\newblock Princeton University Press, 1994.

\bibitem{DBLP:conf/banff/Stirling95}
Colin Stirling.
\newblock Modal and temporal logics for processes.
\newblock In Faron Moller and Graham~M. Birtwistle, editors, {\em Banff Higher
  Order Workshop}, volume 1043 of {\em {Lect. Notes Comput. Sci.}}, pages
  149--237. {Springer-Verlag}, 1995.

\bibitem{DBLP:journals/jlp/ThraneFL10}
Claus Thrane, Uli Fahrenberg, and Kim~G. Larsen.
\newblock Quantitative analysis of weighted transition systems.
\newblock {\em {J. Log. Alg. Prog.}}, 79(7):689--703, 2010.

\bibitem{DBLP:journals/tcs/Breugel01}
Franck van Breugel.
\newblock An introduction to metric semantics: operational and denotational
  models for programming and specification languages.
\newblock {\em {Theor. Comput. Sci.}}, 258(1-2):1--98, 2001.

\bibitem{DBLP:journals/tcs/BreugelW05}
Franck van Breugel and James Worrell.
\newblock A behavioural pseudometric for probabilistic transition systems.
\newblock {\em {Theor. Comput. Sci.}}, 331(1):115--142, 2005.

\bibitem{DBLP:conf/birthday/BreugelW14}
Franck van Breugel and James Worrell.
\newblock The complexity of computing a bisimilarity pseudometric on
  probabilistic automata.
\newblock In Franck van Breugel, Elham Kashefi, Catuscia Palamidessi, and Jan
  Rutten, editors, {\em Horizons of the Mind. {A} Tribute to Prakash
  Panangaden}, volume 8464 of {\em {Lect. Notes Comput. Sci.}}, pages 191--213.
  {Springer-Verlag}, 2014.

\bibitem{inbook/hpa/Glabbeek01}
Rob~J. van Glabbeek.
\newblock The linear time -- branching time spectrum {I}.
\newblock In Jan~A. Bergstra, Alban Ponse, and Scott~A. Smolka, editors, {\em
  Handbook of Process Algebra}, pages 3--99. Elsevier, 2001.

\bibitem{DBLP:conf/fm/WangME93}
Farn Wang, Aloysius~K. Mok, and E.~Allen Emerson.
\newblock Symbolic model checking for distributed real-time systems.
\newblock In Jim Woodcock and Peter~Gorm Larsen, editors, {\em FME}, volume 670
  of {\em {Lect. Notes Comput. Sci.}}, pages 632--651. {Springer-Verlag}, 1993.

\bibitem{DBLP:journals/tcs/ZwickP96}
Uri Zwick and Mike Paterson.
\newblock The complexity of mean payoff games on graphs.
\newblock {\em {Theor. Comput. Sci.}}, 158(1{\&}2):343--359, 1996.
\end{thebibliography}

\end{document}